\documentclass[12pt]{amsart}

\usepackage{amssymb}
\usepackage{textcomp}
\usepackage{color}
\usepackage{comment}
\usepackage{graphicx}

\usepackage{enumitem}

\usepackage{ljm-auth}

\usepackage{microtype} 
\usepackage{tikz}
\usepackage{calc}
\usepackage{comment}
\usepackage{subcaption}
\usepackage{hyperref}

\usetikzlibrary{shapes,positioning,chains,fit,calc,decorations.markings,decorations.pathreplacing}
\tikzstyle{ccyan}=[circle, draw, thick,fill=cyan!30, minimum size=12pt,inner sep=0pt]
\tikzstyle{cgrey}=[circle, draw, thick,fill=gray!30, minimum size=10pt,inner sep=0pt]
\tikzstyle{cgreys}=[circle, draw, thick,fill=gray!30, minimum size=12pt,inner sep=0pt]

\DeclareRobustCommand*\cal{\@fontswitch\relax\mathcal}

\newcommand{\SAF}{\mathtt{2\mbox{-}SAF_t}}
\newcommand{\USAF}{\mathtt{2\mbox{-}USAF_t}}

\newcommand{\da}{2DA_n}
\newcommand{\datheta}{2DA^{\Theta}_n}
\newcommand{\na}{2NA_n}
\newcommand{\natheta}{2NA^{\Theta}_n}
\newcommand{\pa}{2PA_n}
\newcommand{\patheta}{2PA^{\Theta}_n}

\newcommand{\dsize}[1]{\mathsf{2DSIZE( \mathrm{#1} )}}
\newcommand{\dsizetheta}[1]{\mathsf{2D \Theta SIZE( \mathrm{#1} )}}
\newcommand{\nsize}[1]{\mathsf{2NSIZE( \mathrm{#1} )}}
\newcommand{\nsizetheta}[1]{\mathsf{2N \Theta SIZE( \mathrm{#1} )}}
\newcommand{\psize}[1]{\mathsf{2BSIZE( \mathrm{#1} )}}
\newcommand{\psizetheta}[1]{\mathsf{2B \Theta SIZE \left( \mathrm{#1} \right)}}

\newcommand{\dfasize}[1]{\mathsf{2DFASIZE( \mathrm{#1} )}}
\newcommand{\nfasize}[1]{\mathsf{2NFASIZE( \mathrm{#1} )}}

\newcommand{\saffunc}[1]{\mathtt{2\mbox{-}SAF}_{#1}}
\newcommand{\usaffunc}[1]{\mathtt{2\mbox{-}USAF}_{#1}}
\newcommand{\usaflang}[1]{\mathtt{L2USAF}_{#1}}

\newcommand{\ceil}[1]{\lceil #1 \rceil}

\newcommand{\floor}[1]{\lfloor #1 \rfloor}
\newcommand{\Floor}[1]{\left\lfloor #1 \right\rfloor}
\newcommand{\paran}[1]{ \left( #1 \right) }

\newtheorem{THM}{Theorem}

\newtheorem{CR}[THM]{Corollary}

\makeatletter
\newtheorem*{rep@theorem}{\rep@title}
\newcommand{\newreptheorem}[2]{%
\newenvironment{rep#1}[1]{%
 \def\rep@title{#2 \ref{##1}}%
 \begin{rep@theorem}}%
 {\end{rep@theorem}}}
\makeatother

\newtheorem{theorem}{Theorem}
\newtheorem{corollary}{Corollary}
\newtheorem{lemma}{Lemma}
\newreptheorem{lemma}{Lemma}

\author{Kamil Khadiev$^{1,2}$, Rishat Ibrahimov$^{2}$, and Abuzer Yakary{\i}lmaz$^{1}$}
\crauthor{K.\,Khadiev, R.\,Ibrahimov, A.\,Yakary{\i}lmaz} 

\tit{New Size Hierarchies for Two Way Automata}
\shorttit{New Size Hierarchies for Two Way Automata} 

\begin{document}

\maketit
\address{$^1$Center for Quantum Computer Science, Faculty of Computing, University of Latvia
Raina bulv. 19, R\={\i}ga, LV-1586, Latvia\\
$^2$ Kazan Federal University, Kremlevskaya Str. 18, 420008, Kazan, Russia}
\email{kamilhadi@gmail.com, rishat.ibrahimov@yandex.ru, abuzer@lu.lv}
\abstract{We introduce a new type of nonuniform two--way automaton that can use a different transition function for each tape square. We also enhance this model by allowing to shuffle the given input at the beginning of the computation. Then we present some hierarchy and incomparability results on the number of states for the types of deterministic, nondeterministic, and bounded-error probabilistic models. For this purpose, we provide some lower bounds for all three models based on the numbers of subfunctions and we define two witness functions.
} \notes{0}{

}


\section{Introduction}

Nonuniform models (like circuits, branching programs, uniform models using advice, etc.) have played significant roles in computational complexity, and, naturally they have also been investigated in automata theory (e.g. \cite{BST90,Hol02,Kap09,Bal13}). The main computational resource for nonuniform automata is the number of internal states that depends on the input size. Thus we can define linear, polynomial, or exponential size automata models. In this way, for example, nonuniform models allow us to formulate the analog of ``$ \sf P $ versus $  \sf NP $ problem'' in automata theory: Sakoda and Sipser \cite{SS87} conjectured that simulating a two--way nondeterministic automaton by two--way deterministic automata requires exponential number of states in the worst case. But, the best known separation is only quadratic ($O(n^2)$) \cite{C86,Kap05} and the researchers have succeeded to obtain slightly better bounds only for some modified models (e.g. \cite{L01,K11,GGP12,KKM12}). 
Researchers also considered similar question for OBDD model that can be seen as nonuniform automata (e.g. \cite{K16}, \cite{kk2017}, \cite{akk2017}, \cite{agky14}, \cite{agky16}, \cite{ki2017}).

In this paper, we present some hierarchy results for deterministic, nondeterministic, and bounded-error probabilistic nonuniform two--way automata models, which can also be seen as a ``two--way'' version of ordered binary decision diagrams (OBDDs) \cite{Weg00}. For each input length ($n$), our models can have different number of states, and, like Branching programs or the data-independent models defined by Holzer \cite{Hol02}, the transition functions can be changed during the computation. Holzer's model can use a different transition function for each step. We restrict this property so that the transition function is the same for the same tape positions, and so, we can have at most $n$ different transition functions. Moreover, we enhance our models by shuffling the input symbols at the beginning of the computation. We give the definitions and related complexity measures in Section \ref{sec:definitions}.

In order to obtain our main results, we start with presenting some generic lower bounds (Section \ref{sec:lboundS}) by using the techniques given in \cite{Sh59} and \cite{DS90}. Then, we define two witness Boolean functions in Section \ref{sec:functions}: {\em Shuffled Address Function}, denoted $ \SAF $, which is a modification of Boolean functions given in \cite{nw91} (see also \cite{ak96}, \cite{bssw96}, \cite{Kap05}, \cite{K15}), and its uniform version $\USAF$. Moreover, regarding these functions, we provide two deterministic algorithms. In our results, we also use the well known {\em Equality function} $ \mathsf{EQ(X)}=\bigvee_{0}^{\lfloor n/2\rfloor-1} x_i=x_{i+\lfloor n/2\rfloor}$.

In Sections \ref{sec:hierarchy} and \ref{sec:incomparability}, we present our main results based on the size (the number of states) of models. We obtain linear size separations for deterministic models and quadratic size separations for nondeterministic and probabilistic models. Moreover, we investigate the effect of shuffling for all three types of models, and, we show that in some cases shuffling can save huge amount of states and in some other cases shuffling cannot be size efficient. We also show that the constant number of states does not increase the computational power of deterministic and nondeterministic nonuniform models without shuffling. 

\section{Definitions}
\label{sec:definitions}

Our alphabet is binary, $\Sigma=\{0,1\}$. We mainly use the terminologies of Branching programs: Our decision problems are solving/computing Boolean functions: The automaton solving a function accepts the inputs where the function gets the value true and rejects the inputs where the function gets the value false. For uniform models, on the other hand, our decision problems are recognizing languages: The automaton recognizing a language accepts any member and rejects any non-member.

A nonuniform head--position--dependent two--way deterministic automaton working on the inputs of length/size $n \geq 0$ (2DA$_n$) $D_n$ is a 6-tuple
\[
	D_n = (\Sigma,S,s_1,\delta=\{\delta_1,\ldots,\delta_n\},s_a,s_r),
\]
where (i) $ S = \{ s_1,\ldots,s_d \} $ is the set of states ($d$ can be a function in $n$) and $ s_1, s_a, s_r \in S$ ($s_a \neq s_r$) are the initial, accepting, and rejecting states, respectively; and, 
(ii) $\delta$ is a collection of $n$ transition functions such that $ \delta_i: S \setminus \{s_a,s_r\} \times \Sigma \rightarrow S \times \{ \leftarrow,\downarrow,\rightarrow \} $ is the transition function that governs behaviour of $D_n$ when reading the $i$th symbol/variable of the input, where $ 1 \leq i \leq n $. Any given input $u \in \Sigma^n$ is placed on a read-only tape with a single head as $u_1 u_2 \cdots u_{n}$ from the squares 1 to $ |u|=n $, where $u_i \in \Sigma$ is the $i$th symbol of $u$. When $D_n$ is in $s \in S \setminus \{s_a,s_r\}$ and reads $u_i \in \Sigma$ on the tape, it switches to state $ s' \in S $ and updates the head position with respect to $ a \in \{  \leftarrow,\downarrow,\rightarrow  \} $ if $ \delta_i(s,u_i) \rightarrow (s',a) $. If $ a= ``\leftarrow" $ ($``\rightarrow"$), the head moves one square to the left (the right), and, it stays on the same square, otherwise. The transition functions $\delta_1$ and $\delta_n$ must be defined to guarantee that the head never leaves $u$ during the computation. Moreover, the automaton enters $s_a$  or $s_r$ only on the right most symbol and then the input is accepted or rejected, respectively. 

The nondeterministic counterpart of 2DA$_n$, denoted 2NA$_n$, can
choose from more than one transition in each step. So, the range of each transition function is $ \mathcal{P}(S \times \{ \leftarrow,\downarrow,\rightarrow \}) $, where $ \mathcal{P}(\cdot) $ is the power set of any given set. Therefore, a 2NA$_n$ can follow more than one computational path and the input is accepted only if one of them ends with the decision of ``acceptance''. 
Note that some paths end without any decision since the transition function can yield the empty set for some transitions. 

The probabilistic counterpart of 2DA$_n$, denoted 2PA$_n$, is a 2NA$_n$ such that each transition is associated with a probability. Thus, 2PA$_n$s can be in a probability distribution over the deterministic configurations (the state and the position of head forms a configuration) during the computation. To be a well-formed machine, the total probability must be 1, i.e. the probability of outgoing transitions from a single configuration must be always 1. Thus, each input is accepted and rejected by a 2PA$_n$ with some probabilities. An input is said to be accepted/rejected by a (bounded-error) 2PA$_n$ if the accepting/rejecting probability by the machine is at least $1/2 + \varepsilon$ for some $ \varepsilon \in (0,1/2] $.

A function $f_n : \Sigma^n \rightarrow \Sigma $ is said to be computed by a 2DA$_n$ $D_n$ (a 2NA$_n$ $N_n$, a 2PA$_n$ $P_n$) if each member of $ f_n^{-1}(1) $ is accepted by $D_n$ ($N_n$, $P_n$) and each member of $ f_n^{-1}(0) $ is rejected by $D_n$ ($N_n$, $P_n$). 

The class $ \mathsf{2DSIZE(d(n))} $ is formed by the functions $ f = \{ f_0,f_1,f_2,\ldots \} $ such that each $f_i$ is computed by a 2DA$_i$ $ D_i $, the number of states of which is no more than $ d(i) $, where $i$ is a non-negative integer. We can similarly define nondeterministic and probabilistic counterparts of this class, denoted $ \mathsf{2NSIZE(d(n))} $ and $ \mathsf{2PSIZE(d(n))} $ respectively.

We also introduce a generalization of our nonuniform models that can shuffle the input at the beginning of the computation with respect to a permutation. A nonuniform head--position--dependent \textit{shuffling} two--way deterministic automaton working on the inputs of length/size of $n \geq 0$ (2DA$^{\Theta}_n$), say $D^{\theta}_n$, is a 2DA$_n$ that shuffles the symbols of input with respect to $ \theta $, a permutation of $ \{1,\ldots,n\} $, i.e. the $j$-th symbol of the input is placed on $ \theta(j) $-th place on the tape ($ 1 \leq j \leq n $), and then execute the 2DA$_n$ algorithm on this new input. The nondeterministic and probabilistic models can be respectively abbreviated as 2NA$^{\Theta}_n$ and 2PA$^{\Theta}_n$. 

The class $ \mathsf{2D \Theta SIZE(d(n))} $ is formed by the functions $ f = \{ f_0,f_1,f_2,\ldots \} $ such that each $f_i$ is computed by a 2DA$^\Theta_i$ $ D^\theta_i $ whose number of states is no more than $ d(i) $, where $i$ is a non-negative integer and $\theta$ is a permutation of $ \{1,\ldots,n\} $. The nondeterministic and probabilistic classes are respectively represented by $ \mathsf{2N\Theta SIZE(d(n))} $ and $ \mathsf{2P\Theta SIZE(d(n))} $.

Moreover we consider uniform versions of two-way automata, respectively 2DFA and 2NFA. We can define 2DFA in the same way as 2DA$_n$, but it is identical for all $n$, $|S|=const$ and $\delta_i=\delta$ for any $i\in \{1,\ldots,n\}$. Moreover, they can use end-markers, between which the given input is placed on the input tape. We can define 2NFA similarly. The corresponding classes of languages defined by 2DFAs and 2NFAs of size $d$ are denoted $\mathsf{2DFASIZE(d)} $ and $\mathsf{2NFASIZE(d)} $, respectively.

\section{Lower bounds, Boolean functions, and algorithms}
\label{sec:lboundS}

Our key complexity measure behind our results is the number of subfunctions for a given function. It can be seen as the counterpart of ``the equivalence classes of a language'' with respect to  Myhill-Nerode Theorem \cite{RS59}. 

Let $ f $ be a Boolean function defined on $ X = \{x_1,\ldots,x_n\} $. We define the set of all permutations of $ (1,\ldots,n) $ as $ \Theta(n) $. Let $ \theta \in \Theta(n) $ be a permutation. We can order the elements of $ X $ with respect to $ \theta $, say $ (x'_1,\ldots,x'_n) $, and then we can split them into two disjoint non-empty (ordered) sets by picking an index $ i \in \{1,\ldots,n-1\} $: $ X_A = (x'_1,\ldots,x'_i) $ and $ X_B = (x'_{i+1},\ldots,x'_n) $. Let $ \rho $ be a mapping $ \{ x'_1,\ldots,x'_i \} \rightarrow \{ 0,1 \}^i $ that assigns a value to each $ x'_j \in (x'_1,\ldots,x'_i) $. Then, we define function $ f|_{\rho} : X_B \rightarrow \{ 0,1 \} $ that returns the value of $ f $ where the values of the input from $ X_A $ are fixed by $ \rho $. The function $ f|_{\rho} $ is called a subfunction.

The total number of different subfunctions with respect to $ \theta $ and $ i $ is denoted by $ N_i^\theta (f) $. Then, we focus on the maximum value by considering all possible indices:
\[
	N^\theta(f) = \max_{ i \in \{1,\ldots,n-1\} } N_i^\theta (f).
\]
After this, we focus on the best permutation that minimizes the number of subfunctions:
\[
	N(f) = \min_{\theta \in \Theta(n)} N^\theta (f).
\]

Now, we represent the relation between the number of subfunctions for Boolean function and the number of  equivalence classes of a language.

Let $ L $ be a language defined on $ \Sigma = \{0,1\} $. For a given non-negative integer $ n $, $ L_n $ is the language composed by all members of $ L $ with length $ n $, i.e. $ L_n = L \cap \Sigma^n = \{ w \mid w \in L, |w| = n \} $. For $ r < n $, two strings $ u \in \Sigma^r $ and $ v \in \Sigma^r $ are said to be equivalent if for any $ y \in \Sigma^{n-r} $, $ uy \in L_n $ if and only if $ vy \in L_n $. 

We denote the number of non-equivalent strings of length $ r $ as $ R^r(L_n) $. Then, similar to the number of subfunctions,
\[
	R(L_n)  = \max_{r \in \{1,\ldots,n-1\}} R^r(L_n) 
	\mbox{ and }
	R_n(L) = R(L_n).
\]

The function $ f_n^L(X) $ denotes the characteristic Boolean function for language $ L_n $. Thus, we can say that
\[
	N^{id}(f_n^L) = R_n(L).
\]
for $id=(1,\dots,n)$ is natural order.


\subsection{Lower bounds}

First we give our lower bounds on the sizes of models in terms of $ N(f) $. Note that all of lower bounds for shuffling models are valid also for non-shuffling models.

\begin{theorem}
\label{thm:lower-det}
If the function $f(X)$ is computed  by a 2DA$^\Theta_n$ of size $d$ for some permutations $\theta$, then
\[
	N(f) \leq (d+1)^{d+1}. 
\]
\end{theorem}
\begin{proof}
	This result is easily obtained by using the standard and well-known conversion given by Shepherdson \cite{Sh59}. 
\end{proof}

\begin{CR}
\label{cor:lower-2dfa}
If the language $L$ is recognized  by a 2DFA of size $d$, then 
\[
	R_n(L) \leq (d+1)^{d+1},
\]
\end{CR}

\begin{theorem}
\label{thm:lower-non}
If the function $f(X)$ is computed  by a  2NA$^\Theta_n  $ of size $d$ for some permutations $\theta$, then
\[
	N(f) \leq 2^{(d+1)^2}.
\]
\end{theorem}
\begin{proof}
	This result follows from \cite{Var89}. 
\end{proof}

\begin{CR}
\label{cor:lower-2nfa}
If the language $L$ is recognized  by a 2NFA of size $d$, then 
\[
	R_n(L) \leq 2^{(d+1)^2}.
\]
\end{CR}
Based on Theorems \ref{thm:lower-det} and \ref{thm:lower-non}, we can obtain the following result.
\begin{theorem}
\label{thm:regular}
	For constant integer $d$, $ \mathsf{2DSIZE(d)} $ and $ \mathsf{2NSIZE(d)} $ contain only characteristic functions of regular languages.
\end{theorem}
\begin{proof}
	If  $d$ is constant and $\da$ (or $ \na $) $A_n$ computes $f_n$ then $N (f_n)$ is constant and it is same for each $n$. Hence the $R(L_n)$ of the corresponding language $L_n$ is constant too, so the language $L=\bigcup L_n$ is regular since the number of equivalence classes is finite with respect to Myhill-Nerode Theorem \cite{RS59}.
\end{proof}

\begin{theorem}
\label{thm:lower-pro}
If the function $f(X)$ is computed  by a 2PA$^\Theta_n  $ of size $d$ for some permutations $\theta$ with expected running time $T$ and error probability $\varepsilon$, then
\[
	N(f) \leq \left \lceil \frac{4d \left(8 + 3\log T \right)}{\log \left( 1 + 2\varepsilon \right) \left( 1 + \varepsilon \right)} \right \rceil ^{(d+1)^2}.
\]
\begin{proof}
	This result follows from the techniques given in \cite{DS90}.
    
We need some additional definitions to present our lower bound for 2PA$^\Theta_n$.
Let $\beta \geq 1$. Two numbers $p$ and $p'$ are said to be $\beta - close$ if either
\begin{itemize}
\item $p = p' = 0$ or
\item $p \geq 0$, $p' \geq 0$, and $\beta ^{-1} \leq p / p' \leq \beta$.
\end{itemize} 
Two numbers $p$ and $p'$ are said to be $\beta - close$ mod $\lambda$ if either
\begin{itemize}
\item $p \leq \lambda$ and $p' \leq \lambda$ or
\item $p \geq \lambda$, $p' \geq \lambda$ and $p$ and $p'$ are $\beta-close$.
\end{itemize} 

Let $P = {\{ p_{i,j} \} }^{m}_{i,j = 1}$ be an $m$-state Markov chain with starting state $1$ and two absorbing states $m-1$ and $m$; $a(P)$ denote the probability that Markov chain $P$ is absorbed in state $m$ when started in state $1$; and, $T(P)$ denote the expected time to absorption into one of the states $m-1$ or $m$.
Two Markov chains $P = {\{ p_{i,j} \} }^{m}_{i,j = 1}$ and $P' = {\{ p'_{i,j} \} }^{m}_{i,j = 1}$ are said to be $\beta - close$ mod $\lambda$ if, for each pair $i, j$, $p_{i,j}$ and $p'_{i,j}$ are $\beta - close$ mod $\lambda$.

Let $\pi\in \Pi(\theta)$ be any partition such that $\theta\in \Theta(n)$ is the order of the inputs for $A^\theta_n$, and,  $\pi=(X_A,X_B)$ and $|X_A|=u$. 

Let us consider configurations $b_i$ for $i\in\{0,\dots,2d+2\}$. Configuration $b_0$ is initial configuration of the automata, $b_{d+1}$ is accepting state and position of head on the last symbol, $b_{d+2}$ is similar but in rejecting state. For $i\in\{1,\dots,d\}$, configuration $b_i$ is for position $u$ and state of the automata $i$. For $i\in\{d+3,\dots,2d+2\}$, configuration $b_i$ is for position $u+1$ and state of the automata $i-d-3$. 
We will use three object for describing computational process for the automata: matrix $M_A(\sigma,\gamma)$ and vectors $p^0,q$.
Matrix $M_A(\sigma,\gamma)$ is $(2d+3)\times (2d+3)$ block diagonal matrix with two blocks $M_A(\sigma)$ and $M_A(\gamma)$. $i$-th line and row of the matrix  $M_A(\sigma,\gamma)$ corresponding to configuration $b_i$. Matrices $M_A(\sigma)=\{m_{i,j}\}$, $M_A(\sigma)=\{m'_{i,j}\}$ have following elements:
\begin{itemize}
\item If $A^\theta_n$ begins the computation from the configuration $b_i$, then it reaches the configuration $b_{j+d+3}$ early than another $b_r$ for  $r\geq d+3$  within probability $m_{ij}$

\item if $A^\theta_n$ begins the computation from the configuration $b_{i+d+3}$, then it reaches the configuration 
$b_{j}$ early than another $b_r$ for $r\leq d$ within probability $m'_{ij}$.
\end{itemize}
The vector $p^0$ and $q$ have size $(2d+3)$ and $i$-th element of vectors also correspond to $b_i$. The vector $p^0$ represents initial distribution of probability for configuration of the automata. So $p^0_0=1$ and other elements are $0$. The vector $q$ is characteristic vector of accepting state. So,  $q_{b+1}=1$ and other elements are $0$

Let us discuss some properties.

\begin{lemma}
\label{lm-linar3}
If the function $f(X)$ is computed  by $A^\theta_n$ of size $d$ within error probability $\varepsilon$, then for any $\nu\in\{0,1\}^n$ satisfying $f(\nu)=1$, there is a $t'$ such that
\[
	p^0 \cdot\Big(M_A(\sigma,\gamma)\Big)^{t'}\cdot q \geq \frac{1}{2} + \varepsilon
\]
On the other hand, for any $\nu\in\{0,1\}^n$ satisfying $f(\nu)=0$ there is no such $t'$.
\end{lemma}
\begin{proof}
	Since computation is probabilistic, there can be more than one path from the configuration $b_i$, where it reaches the configuration $b_{j+d+3}$ early than another $b_r$ for  $r\geq d+3$. But by construction of matrix $M_A$, we consider all of them. In \cite{DS90} have shown that we can model probabilistic computation in this way.
\end{proof}

We have shown that we model computation of 2PA$^\Theta_n  $ by Markov chain specified by matrix $M_A(\sigma,\gamma)$.

\begin{lemma}
	\label{lm-betaclose}
	Let $P$ and $P'$ are two $m$-state Markov chains and $a\left(P\right) \geq \frac{1}{2} + \varepsilon$. 
	Let $T = max\left(T\left(P\right), T\left(P'\right), m\right)$, $\lambda = {\varepsilon ^ 2}/{256T^3}$ и $\beta = \sqrt[2m]{\frac{1 + \varepsilon + \varepsilon ^ 2}{1 + \varepsilon}}$. If $P$ and $P'$ are $\beta-close$ mod $\lambda$, then $a\left(P'\right) \geq \frac{1}{2} + \varepsilon / 4$.
\end{lemma}

\begin{proof}
Dwork and Stockmeyer \cite{DS90} have shown that

$$a\left(P'\right) \geq \left(1 - 2\lambda m^3 \right){\beta}^{-2m}a\left(P\right) - 4\sqrt{\lambda m T}.$$

Therefore the bound on $a\left(P'\right)$ may be obtained by substituting the values of $\lambda$ and $\beta$. 
\end{proof}

Let ${\bf M}$ be the set of all possible matrices $M_A(\sigma)$ such that ${\bf M}=\{M_A(\sigma):\sigma\in\{0,1\}^{|X_A|}\}$ and, for any $P \in {\bf M}$ and $P' \in {\bf M}$, $P$ and $P'$ are $\beta - close$ mod $\lambda$.

The inequality $N^{\pi}(f)\leq |{\bf M}|$ can be obtained in the same way.\footnote{$ M(\sigma) $ includes all information about the behaviour of automaton on $ \sigma $ and so if there are two different $ \sigma $ and $ \sigma' $ with different subfunctions but their matrices $ M(\sigma)$ and $ M(\sigma')$ are $\beta - close$ mod $\lambda$, it should be different on some $ \gamma $. But the behaviour of the automaton on this input is the same and therefore matrices $ M(\sigma,\gamma)$ and $ M(\sigma,\gamma') $ are $\beta - close$ mod $\lambda$.
This is a contradiction.} To estimate $|{\bf M}|$ we use technique similar to one from \cite{DS90}. Let $c = d(d+1)$. Define an equivalence relation on matrices in ${\bf M}$ as follows: 
$Q(w) \equiv Q'(w) \Leftrightarrow \forall i, j \quad (q_{i, j}(w) \leq \lambda \Leftrightarrow q'_{i, j}(w) \leq \lambda )$.
Let $E$ be a largest equivalence class. Since there are at most $2^c$ equivalence classes, $|{\bf M}| \leq 2^c |E(w)|$. Size of $E(w)$ is obtained in \cite{DS90}, since by substituting values $\lambda$ and $\beta$ we have:

$$N^{id}(f) \leq |{\bf M}| \leq 2^c |E(w)| \leq $$

$$\leq \left \lceil \frac{4d(8 + 3\log T)}{\log (1 + 2\varepsilon)/(1 + \varepsilon)} \right \rceil ^ c.$$

The proof of Theorem \ref{thm:lower-pro} is completed.

\end{proof}
\end{theorem}

\begin{corollary}
	\label{cor:lower-pro}
	If the function $f(X)$ is computed  by a 2PA$^\Theta_n  $ of size $d$ for some permutations $\theta$ with expected running time $ T \geq 256$ and error probability $\varepsilon = \frac{1}{5}$, then
	\[
		N(f) \leq (32d\log T)^{(d+1)^2}.
	\]	
\end{corollary}
\begin{proof}
	For $\varepsilon = \frac{1}{5} $, $ \log \left( 1 + 2\varepsilon \right) \left( 1 + \varepsilon \right) > \frac{1}{2} $, and, for $ T \geq 256 $, $ \log T \geq 8 $.
\end{proof}

\subsection{Boolean Functions}

We define two Boolean functions: (1) A modification of Boolean function given in \cite{AK13,K15,nw91,bssw96,Kap05} {\em Shuffled Address Function}, denoted $\saffunc{t}$, and (2) {\em Uniform Shuffled Address Function} $\usaffunc{t}$ as a modification of $\saffunc{t}$. We also use the language $\usaflang{t} $, the characteristic function of which is $\usaffunc{t}$. 

\subsubsection{Boolean Function $ \saffunc{t} $:}
\label{sec:functions}

We divide all input into two parts, and each part into $t$ blocks. Each block has \textit{address} and \textit{value}. Formally, Boolean function
	$ \saffunc{t} (X):\{0,1\}^n\to \{0,1\}$ for integer  $t=t(n)$ such that
\begin{equation}
	\label{kw}
	2t(2t + \lceil \log 2t \rceil)<n.
\end{equation}
We divide the input variables (the symbols of the input) into $2t$ blocks. There are $ \left\lfloor \frac{n}{2t} \right\rfloor =q$ variables in each block.  After that, we divide each block into {\em address} and {\em value} variables (see Figure \ref{fig:address-value}). The first  $\lceil\log 2t\rceil$ variables of block are {\em address} and the other $q-\lceil\log 2t\rceil=b$ variables of block are {\em value}.
We call $x^{p}_{0},\dots,x^{p}_{b-1}$ and  $y^{p}_{0},\dots,y^{p}_{ \lceil\log 2w\rceil}$ are the {\em value} and the  {\em address}  variables of the $p$th block, respectively, for $p\in\{0,\dots,2t-1\}$.

\begin{center}
\begin{figure}[tbh]
    \includegraphics[width=12cm]{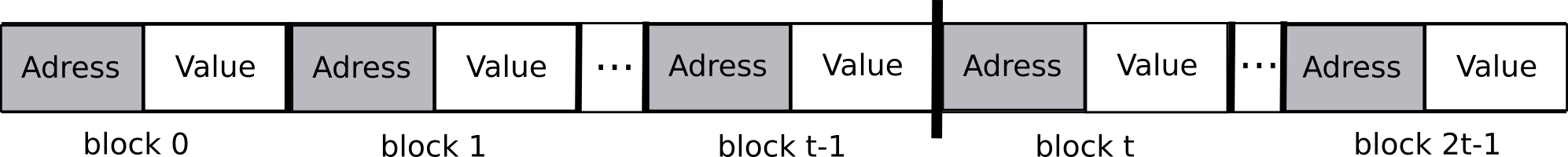}
  \caption{{\em Address} and {\em value} bits of blocks.}
  \label{fig:address-value}
\end{figure}
\end{center}

Function $ \saffunc{ t }(X) $ is calculated based on the following five sub-routines:
\begin{enumerate}
	\footnotesize
	\item $Adr:\{0,1\}^n\times\{0,\dots,2t-1\}\to \{0,\dots,2t-1\}$ gets the address of a block: 
		\[
			Adr(X,p)=\sum_{j=0}^{\lceil\log 2t\rceil-1}y^{p}_{j}\cdot 2^{j} (mod\textrm{ }2t). 
		\]
	\item $Ind:\{0,1\}^n\times\{0,\dots,2t-1\}\to \{-1,\dots,2t-1\}$ gets the number of block by address: 
		\begin{displaymath}
			Ind(X,a) = \left\{ \begin{array}{ccl}
			p & , & \textrm{where $p$ is the minimal number such that $Adr(X,p)=a$}, \\
			-1& , & \textrm{if there are no such $p$}
			\end{array} \right. .
		\end{displaymath}
	\item $Val:\{0,1\}^n\times\{0,\dots,2t-1\}\to \{-1,\dots,t-1\}$ gets the value of the block with address $i$:
		\begin{displaymath}
			Val(X,a) = \left\{ \begin{array}{ccl} 
			\sum_{j=0}^{b-1}x^{p}_{j} (mod\textrm{ }t) &, & \textrm{where }p=Ind(X,a)\textrm{ for $p\geq 0$}, \\
			-1 & , & \textrm{if }Ind(X,a)<0
			\end{array} \right. .
		\end{displaymath}
\end{enumerate}
Suppose that we are at the $i$-th step of iteration.
\begin{enumerate}
	\footnotesize
	\item[4.] $Step_1:\{0,1\}^n\times\{0,\dots,1\}\to \{-1,t\dots,2t-1\}$ gets the first part of the $i$th step of iteration:
		\begin{displaymath}
			Step_1(X,i) = \left\{ \begin{array}{ccl}
			-1 & , & \textrm{if }  Step_2(X,i-1)=-1, \\			
			Val(X,Step_2(X,i-1)) + t & , & \textrm{otherwise}    
			\end{array} \right. .
		\end{displaymath}
	\item[5.] $Step_2:\{0,1\}^n\times\{-1,\dots,1\}\to \{-1,\dots,t-1\}$ gets the second part of the $i$th step of iteration:
		\begin{displaymath}
			Step_2(X,i) = \left\{ \begin{array}{ccl}
			-1  & , & \textrm{if }  Step_1(X,i)=-1, \\
			2 & , & \textrm{if }  i=-1\\
			Val(X,Step_1(X,i)) & , & \textrm{otherwise}  
			\end{array} \right. .
		\end{displaymath}
\end{enumerate}
Function $ \saffunc{t}(X) $ is computed iteratively:
$ 
\saffunc{t}(X) = \left\{ \begin{array}{ll}
	0, & \textrm{if }  Step_2(X,1)\leq 0, \\
	1, & \textrm{otherwise }
	\end{array} \right. :
$
\begin{enumerate}
	\item We find the block with address $2$ in the first part and compute the value of this block, which is the address of the block for the second part.
	\item We take the block from the second part with the computed address and compute value of the block, which is the address of the new block for the first part. 
	\item We find the block with new address in the second part and check value of this block. If the value is greater than $0$, then value of $\saffunc{t}$ is $1$, and $0$ otherwise.   
\end{enumerate}
If we do not find block with searching address $ n $ in any phase then value of $\saffunc{t}$ is also $0$. See the Figure \ref{fig:saf} for the iterations of the function.

 \begin{center}
\begin{figure}[tbh]
    \includegraphics[width=12cm]{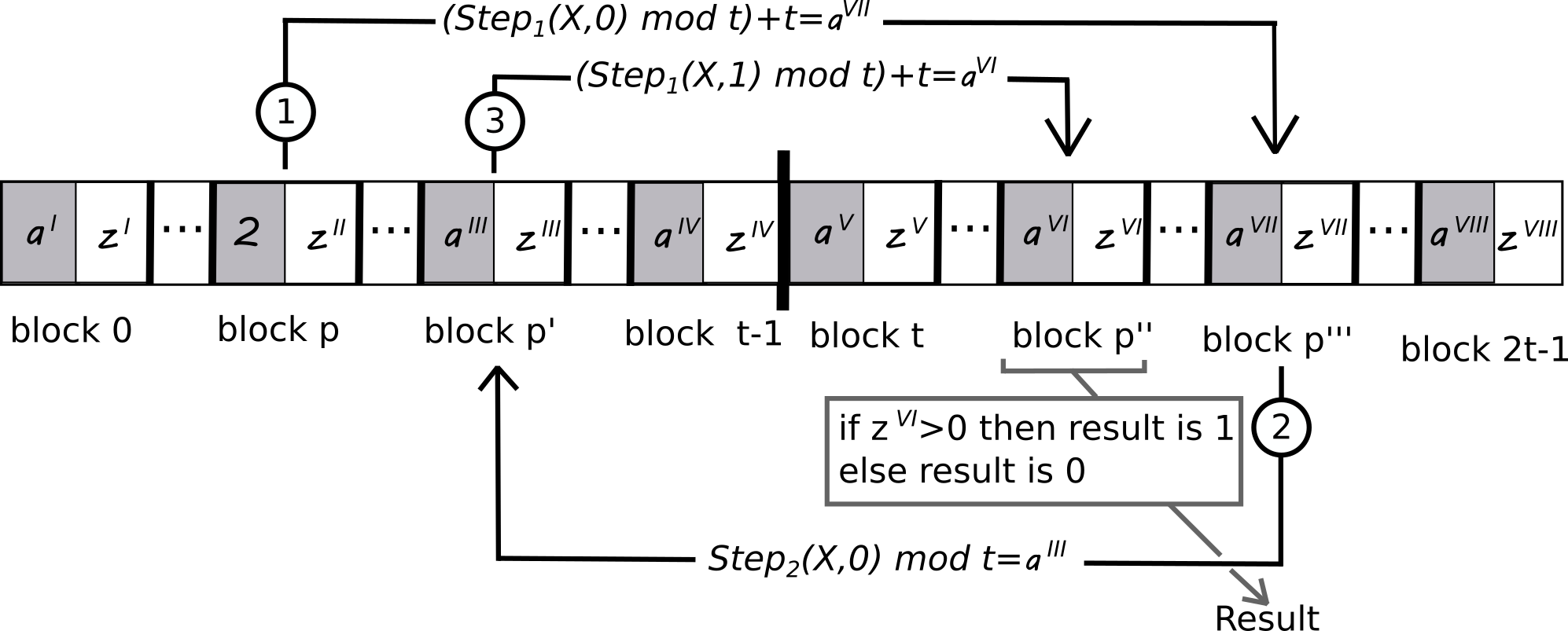}
  \caption{$\saffunc{t}(X)$, $Step_1(X,i)$ and $Step_2(X,i)$ functions. In the picture we write $a$ in address bits of $p$-th block (gray box) if $A(X,p)=a$.}
  \label{fig:saf}
\end{figure}
\end{center}

\begin{theorem}
	\label{thm:lower-2saf}
	For integer  $t=t(n)$, 
	$
		N(\saffunc{t})\geq t^{t-2},
	$
	where $ t $ satisfies $ 2t(2t+\lceil \log 2t \rceil) < n $.
\end{theorem}
\begin{proof}
	 The proof is based on the following two technical Lemmas \ref{good-set} and \ref{good-input}.
\begin{lemma}
\label{good-set}
	Let $t=t(n)$ be some integers satisfying Inequality (\ref{kw}) and $\pi=(X_A,X_B)$ be a partition such that $X_A$ contains at least $t$ {\em value} variables from exactly $t$ blocks. Then, $X_B$ contains at least $t$ {\em value} variables from exactly $t$ blocks.
\end{lemma}
\begin{proof}
We define  $I_A=\{i:$  $X_A$ contains at least $t$ {\em value} variables from $i$th block$\}$. Let $i'\not\in I_A$. Then, $X_A$ contains at most $t-1$ {\em value} variables from $i'$th block, so $X_B$ contains at least $b-(t-1)$ {\em value} variables from $i'$th block. By (\ref{kw}), we can get
\[b-(t-1)= \left\lfloor \frac{n}{2t} \right\rfloor- \lceil \log 2t \rceil-(t-1)\]
which is bigger than
\[(2t + \lceil \log 2t \rceil)- \lceil \log 2t \rceil-(t-1)=2t-(t-1)=t+1.\]
Let $I=\{0,\dots,2t-1\}$ be the numbers of all blocks and $i'\in I\backslash I_A$. Then, we can follow that $|I\backslash I_A|=2t- t =t$.
\end{proof}

Let $\theta\in \Theta(n)$ be any order. Then, we pick a partition $\pi=(X_A,X_B)\in \Pi(\theta)$ such that $X_A$ contains at least $t$ {\em value} variables from exactly $t$ blocks. We define $I_A=\{i:$  $X_A$ contains at least $t$ {\em value} variables from $i$th block$\}$ and $I_B=\{0,\dots,2t-1\}\backslash I_A$. By the proof of Lemma \ref{good-set}, we know that $|I_B|=t$.

Let $(\sigma,\gamma)$ be the partition for the input $\nu$ with respect to $\pi$. We define the sets $\Psi\subset\{{0,1\}^{|X_A|}}$ and $\Gamma\subset\{{0,1\}^{|X_B|}}$ for the input $\nu$ with respect to $\pi$ that satisfies the following conditions. For  $\sigma,\sigma'\in\Psi$, $\gamma\in\Gamma$,  $\nu=(\sigma,\gamma)$, and $\nu'=(\sigma',\gamma)$:
\begin{itemize}
	\item For any $z\in\{0,\dots, t-1\}$, $Ind(\nu,z)\in I_A$; 
	\item For any $z\in\{w,\dots, 2t-1\}$, $Ind(\nu,z)\in I_B$;
	\item There is $z\in\{2,\dots, t-1\}$ such that $Val(\nu',z)\neq Val(\nu,z)$;
	\item The value of $x^{p}_{j}$ is $0$ for any $p\in I_B$ and $x^{p}_{j}\in X_A$; 
	\item The value of $x^{p}_{j}$ is $0$ for any $p\in I_A$ and $x^{p}_{j}\in X_B$;
	\item $ 	Val(\nu,0)=2t-2, Val(\nu',1)=2t-1; \mbox{ and,} $
	\item $ Val(\nu,2t-2)=0,Val(\nu,2t-1)=1.  $
\end{itemize}

\begin{lemma}
	\label{good-input}
	For any sequence $(a_2,\dots,a_{ 2t-3})$, where $a_i\in\{0,\dots, t-1\}$, there are $\sigma\in \Psi$ and $\gamma\in\Gamma$ such that $Val(\nu,z_i)=a_i$ for $z_i\in\{2,\dots, t-1\}$ and $i\in \{2,\dots, 2t-3\}$.
\end{lemma}
\begin{proof}
	Let $p_i\in I_A$ such that $p_i=Ind(\nu,z_i)$ for $i\in \{2,\dots, t-1\}$. Remember that the value of $x^{p_i}_{j}$ is $0$ for any $x^{p_i}_{j}\in X_B$. Hence the value of $Val(\nu,z_i)$ depends only on the variables from $X_A$. At least $t$ value variables of $p_i$th block belong to $X_A$. Hence we  can choose input with $a_i$ $1$'s in the value variables of $p_i$th block which belongs to $X_A$.
For set $\Gamma$ and $i\in \{t,\dots, 2t-3\}$, we can follow the same proof.
\end{proof}

Remember the statement of the theorem: For integer  $t=t(n)$, if $ \saffunc{t} $ satisfies Inequality (\ref{kw}), then
\[
	N(\saffunc{t})\geq t^{t-2}.
\]

Here are the details for the proofs.

Let $\theta\in \Theta(n)$ be an order. Then, we pick the partition $\pi=(X_A,X_B)\in \Pi(\theta)$  such that $X_A$ contains at least $t$ value variables from exactly $t$ blocks.

Let $\sigma,\sigma'\in \Psi$ be two different inputs and $\tau$ and $\tau'$ be their corresponding mappings, respectively. We show that the subfunctions $\SAF|_\tau$ and $\SAF|_{\tau'}$ are different. Let $z\in\{2,\dots, t-1\}$  such that $s'=Val(\nu',z)\neq Val(\nu,z)=s$.

 If $z>2$, then we choose $\gamma\in \Gamma$ providing that $Val(\nu,s+t)=1$, $Val(\nu',s'+t)=0$, and $Val(\nu,r)=Val(\nu',r)=z$, where $r=Val(\nu,2)$. That is,
\begin{itemize}
	\item $Step_1(\nu,0)=Step_1(\nu',0)=r$, 
	\item $Step_2(\nu,0)=s$ and $Step_2(\nu',0)=s'$,
	\item $Step_1(\nu,1)=2t-1$ and $Step_1(\nu',1)=2t-2$, and,
	\item $Step_2(\nu,1)=1$ and $Step_2(\nu',1)=0$.
\end{itemize}
Thus,  $\saffunc{t}(\nu)=1$ and $\saffunc{t}(\nu')=0$.

If $z=2$, then we choose $\gamma\in \Gamma$ providing that $Val(\nu,s+t)=1$ and $Val(\nu',s'+t)=0$. That is, 
\begin{itemize}
	\item $Step_1(\nu,0)=s$ and $Step_1(\nu',0)=s'$,
	\item $Step_2(\nu,1)=1$ and $Step_2(\nu',1)=0$,
	\item $Step_1(\nu,1)=2t-1$ and $Step_1(\nu',1)=2t-2$, and,
	\item $Step_2(\nu,1)=1$ and $Step_2(\nu',1)=0$.
\end{itemize}
Hence $\saffunc{t}(\nu)=1$ and $\saffunc{t}(\nu')=0$.

Therefore $\saffunc{t}|_\tau(\gamma)\neq \saffunc{t}|_{\tau'}(\gamma)$ and also $\saffunc{t}|_\tau\neq \saffunc{t}|_{\tau'}$.

Now, we compute $|\Psi|$. For $\sigma\in\Psi$, we can get each value of $Val(\nu,i)$ for $2\leq i \leq t-1$. It means $|\Psi|\geq t^{t-2}$ due to Lemma \ref{good-input}.
Therefore, $N^{\pi}(\saffunc{t})\geq t^{t-2}$, and, by definition of $N(\saffunc{t})$, we have $N(\saffunc{t})\geq t^{t-2}$.
\end{proof}

\begin{theorem}
	\label{thm:det-alg-2saf}
	There is a $ \da $ $A_n$ of size $13t+4$ that computes $ \saffunc{t} $.
\end{theorem}
\begin{proof}
	Let $ \nu\in\{0,1\}^n$ be the input. We begin with the first part of automaton $A_n$ that computes $Step_1(X,0)$.
Automaton $A_n$ checks each $j$th block for predicate $Adr(\nu, j)=2$. If it is true, then $A_n$ computes $Val(\nu, 2)=r'$, and, it checks the next block, otherwise. If $A_n$ checks all blocks and does not find the block, it switches to the rejecting state. If $A_n$ finds $r'$, then it goes to one of the special state $s'_{r'}$. From this state, the automaton returns back to the beginning of the input.

	We continue with the second part of automaton $A_n$ that computes  $Step_2(X,0)$.
From the state  $s'_{r'}$, $A_n$ checks each $j$th block  for predicate $Adr(\nu, j)=r'$. If it is true, then $A_n$ computes $Val(\nu, r')=r''$, and, it checks the next block, otherwise. If $A_n$ checks all blocks and does not find the block, it switches to the rejecting state. If $A_n$ finds $r''$, then it goes to one of special state $s''_{r''}$. From this state, the automaton returns back to the beginning of the input.

	Now, we describe the third part of automaton $A_n$ that computes  $Step_1(X,1)$. From the state  $s''_{r''}$,  $A_n$ checks each $j$th block for predicate $Adr(\nu, j)=r''$. If it is true, then $A_n$ computes $Val(\nu, r'')=r'''$, and, it checks the next block, otherwise. If $A_n$ checks all blocks and does not find the block, then it switches to the rejecting state.
If $A_n$ finds $r'''$, then it goes to one of special state $s'''_{r'''}$. From this state automaton returns back to the beginning of the input.

	The forth part of automaton $A_n$ computes  $Step_2(X,1)$.
From the state  $s'''_{r'''}$, $A_n$ checks each $j$th block for predicate $Adr(\nu, j)=r'''$. If it is true, then $A_n$ computes $Val(\nu, r''')=r^{IV}$, and, it checks next block otherwise. If $A_n$ checks all blocks and does not find the block,  it switches to the rejecting state.
If $A_n$ finds $r^{IV}$ and $r^{IV}=1$, the automaton accepts the input and rejects the input, otherwise.

	In the first part, the block checking procedure uses only $2$ states. Computing $Val(\nu,2)$ uses $w$ states and there are $t$ $s'_{r'}$ states. So, the size of the first part is $2+2t$. In the second part, the block checking procedure uses only $2$ states and we have $t$ blocks to check the state pairs for each value of $r'$, that is $2t$ states. Computing $Val(\nu,r')$ uses $t$ states and also $A_n$ has $t$ $s''_{r''}$ states. Therefore, the size of the second part is $4t$. Similarly, we can show that the size of the third part is $4t$. In the fourth part, we need $3t$ states for procedure checking and computing $Val(\nu,r''')$. $A_n$ has also one {\em accept} and one {\em reject} states. So, the size of the forth part is $3t+2$. Thus, the overall size of $A_n$ is $13t+4$.
\end{proof}

\subsubsection{Boolean Function $\usaffunc{t}$:}
\label{funcUS}
The definition of $\USAF$ is as follows:
\[ \usaffunc{t} (X):\{0,1\}^n\to \{0,1\} \mbox{ for integer  $t=t(n)$ satisfying that } \]	
\begin{equation}
	\label{kw2}
	4t(2t + \lceil \log 2t \rceil)<n.
\end{equation}
We denote its language version as $ \usaflang{t} $.

	We divide the input variables (the symbols of the input) into $2t$ blocks. There are $ \left\lfloor \frac{n}{2t} \right\rfloor =q$ variables in each block.  After that, we divide each block into {\em mark}, {\em address}, and {\em value} variables. All variables that are in odd positions are {\em mark}. The type of the bit on an even position is determined by the value of previous {\em mark} bit.	The bit is {\em address}, if the previous {\em mark} bit's value is 0, and {\em value} otherwise. The first  $\lceil\log 2t\rceil = c$ variables of block that are in the even positions denote the {\em address} and the other $q/2-\lceil\log 2t\rceil=b$ variables of block denote the {\em value}.

We call $x^{p}_{0},\dots,x^{p}_{b-1}$, $y^{p}_{0},\dots,y^{p}_{ \lceil\log 2t\rceil}$ 
and $z^{p}_{0},\dots,z^{p}_{q/2}$ are the {\em value}, {\em address}, and the  {\em mark} variables of the $p$th block, respectively, for $p\in\{0,\dots,2t-1\}$.

Function $ \usaffunc{t}(X)$ is calculated based on the following five sub-routines:
\begin{enumerate}
	\footnotesize
	\item $Adr:\{0,1\}^n\times\{0,\dots,2t-1\}\to \{0,\dots,2t-1\}$ gets the address of a block: 
		\[
			Adr(X,p)=\sum_{j=0}^{c-1}y^{p}_{j}\cdot 2^{c - j - 1} (mod\textrm{ }2t). 
		\]
	\item $Ind:\{0,1\}^n\times\{0,\dots,2t-1\}\to \{-1,\dots,2t-1\}$ gets the number of block by address: 
		\begin{displaymath}
			Ind(X,a) = \left\{ \begin{array}{ccl}
			p & , & \textrm{where $p$ is the minimal number such that $Adr(X,p)=a$}, \\
			-1& , & \textrm{if there are no such $p$}
			\end{array} \right. .
		\end{displaymath}
	\item $Val:\{0,1\}^n\times\{0,\dots,2t-1\}\to \{-1,\dots,t-1\}$ gets the value of the block with address $i$:
		\begin{displaymath}
			Val(X,a) = \left\{ \begin{array}{ccl} 
			\sum_{j=0}^{b-1}x^{p}_{j} (mod\textrm{ }t) &, & \textrm{where }p=Ind(X,a)\textrm{ for $p\geq 0$}, \\
			-1 & , & \textrm{if }Ind(X,i)<0
			\end{array} \right. .
		\end{displaymath}
\end{enumerate}
Suppose that we are at the $i$-th step of iteration.
\begin{enumerate}
	\footnotesize
	\item[4.] $Step_1:\{0,1\}^n\times\{0,\dots,1\}\to \{-1,t\dots,2t-1\}$ gets the first part of the $i$th step of iteration:
		\begin{displaymath}
			Step_1(X,i) = \left\{ \begin{array}{ccl}
			-1 & , & \textrm{if }  Step_2(X,i-1)=-1, \\			
			Val(X,Step_2(X,i-1)) + t & , & \textrm{otherwise}    
			\end{array} \right. .
		\end{displaymath}
	\item[5.] $Step_2:\{0,1\}^n\times\{-1,\dots,1\}\to \{-1,\dots,t-1\}$ gets the second part of the $i$th step of iteration:
		\begin{displaymath}
			Step_2(X,i) = \left\{ \begin{array}{ccl}
			-1  & , & \textrm{if }  Step_1(X,i)=-1, \\
			2 & , & \textrm{if }  i=-1\\
			Val(X,Step_1(X,i)) & , & \textrm{otherwise}  
			\end{array} \right. .
		\end{displaymath}
\end{enumerate}
Remark that the address of the current block is computed on the previous step. Function $\usaffunc{t}(X)$ is computed as:
$
	\usaffunc{t}(X) = \left\{ \begin{array}{ll}
	0, & \textrm{if }  Step_2(X,1)\leq 0, \\
	1, & \textrm{otherwise }
	\end{array} \right. .
$

\begin{theorem} 
	\label{thm:lower-2usaf}
	For integer  $t=t(n)$, 
	$
		R_n(\usaflang{t})\geq t^{t-2},
	$
	where $ t $ satisfies $ 4t(2t + \lceil \log 2t \rceil ) < n $.
\end{theorem}
\begin{proof}
	We pick the partition $\pi=(X_A,X_B)\in \Pi(id)$ such that $X_A$ contains exactly $w$ blocks.

Let $\sigma,\sigma'\in \Psi$ be two different inputs and $\tau$ and $\tau'$ be their corresponding mappings, respectively. We show that the subfunctions $\USAF|_\tau$ and $\USAF|_{\tau'}$ are different. Let $z\in\{2,\dots, t-1\}$  such that $s'=Val(\nu',z)\neq Val(\nu,z)=s$.

 If $z>2$, then we choose $\gamma\in \Gamma$ providing that $Val(\nu,s+t)=1$, $Val(\nu',s'+t)=0$, and $Val(\nu,r)=Val(\nu',r)=z$, where $r=Val(\nu,2)$. That is,
\begin{itemize}
	\item $Step_1(\nu,0)=Step_1(\nu',0)=r$, 
	\item $Step_2(\nu,0)=s$ and $Step_2(\nu',0)=s'$,
	\item $Step_1(\nu,1)=2t-1$ and $Step_1(\nu',1)=2t-2$, and,
	\item $Step_2(\nu,1)=1$ and $Step_2(\nu',1)=0$.
\end{itemize}
Thus,  $\USAF(\nu)=1$ and $\USAF(\nu')=0$.

If $z=2$, then we choose $\gamma\in \Gamma$ providing that $Val(\nu,s+t)=1$ and $Val(\nu',s'+t)=0$. That is, 
\begin{itemize}
	\item $Step_1(\nu,0)=s$ and $Step_1(\nu',0)=s'$,
	\item $Step_2(\nu,1)=1$ and $Step_2(\nu',1)=0$,
	\item $Step_1(\nu,1)=2t-1$ and $Step_1(\nu',1)=2t-2$, and,
	\item $Step_2(\nu,1)=1$ and $Step_2(\nu',1)=0$.
\end{itemize}
Hence $\USAF(\nu)=1$ and $\USAF(\nu')=0$.

Therefore $\USAF|_\tau(\gamma)\neq \USAF|_{\tau'}(\gamma)$ and also $\USAF|_\tau\neq \USAF|_{\tau'}$.

Now, we compute $|\Psi|$. For $\sigma\in\Psi$, we can get each value of $Val(\nu,i)$ for $2\leq i \leq t-1$. It means $|\Psi|\geq t^{t-2}$ due to Lemma \ref{good-input}.
Therefore, $N^{\pi}(\USAF)\geq t^{t-2}$, and, by definition of $N(\USAF)$, we have $N(\USAF)\geq t^{t-2}$.
\end{proof}

\begin{theorem}
	\label{thm:det-alg-2usaf}
	There is a $ 2DFA $ $A_n$ of size $23t + 2(1 + 3t)\log{t} +6$ recognizing $ \usaflang{t} $. 
\end{theorem}
\begin{proof}
	The computation of $A_n$ consists of four parts similar to the automaton given in Theorem \ref{thm:det-alg-2saf}.
    
In the first part, the block checking procedure uses $2\log{2t} + 2$ states. Computing $Val(\nu,2)$ uses $2t$ states and there are $w$ $s'_{r'}$ states. So, the size of the first part is $3t + 2\log{2t} + 2$. In the second part, the block checking procedure uses $2\log{2t} + 2$ states and we have $t$ blocks to check the state pairs for each value of $r'$, that is $w(2\log{2t} + 2)$ states. Computing $Val(\nu,r')$ uses $2t$ states and also $A_n$ has $t$ $s''_{r''}$ states. Therefore, the size of the second part is $5t + 2t\log{2t}$. Similarly, we can show that the size of the third part is $5t + 2t\log{2t}$. 
In the fourth part, we use $4t + 2t\log{2t}$ states for procedure checking and computing $Val(\nu,r''')$. $A_n$ has also one {\em accept} and one {\em reject} states. So, the size of the forth part is $4t + 2t\log{2t}+2$. Thus, the overall size of $A_n$ is $23t + 2(1 + 3t)\log{t} +6$. 
\end{proof}
We can simplify the formula for the number of states for bigger $ t $ values, e.g.:
	\begin{itemize}
		\item For $ t > 2906 $, there is a $ 2DFA $ of size $ 8t\log(t) $ recognizing $ \usaflang{t} $.
		\item For $ t > 26 $, there is a $ 2DFA $ of size $ 11t\log(t)  $ recognizing $ \usaflang{t} $.
		\item For $ t > 3 $, there is a $ 2DFA $ of size $ 19t\log(t) $ recognizing $ \usaflang{t} $.
	 \end{itemize} 

\section{Hierarchies results}
\label{sec:hierarchy}

Now, we can follow our hierarchies and incomparability results.

\begin{theorem}
\label{thm:hieararchy-det}
	Let $ d: \mathbb{N} \rightarrow \mathbb{N} $ be a function satisfying $ 13d + 43 < \sqrt{\frac{n}{6}}$. Then, 
	$
		\dsizetheta{d} \subsetneq \dsizetheta{13d+43}.
	$
\end{theorem}
\begin{proof}
	By Theorem \ref{thm:lower-2saf}, we can follow that
	$
		(d+1)^{d+1} < (d+3)^{d+1} \leq N( \saffunc{d+3} ).
	$
	Suppose that there is a $ \datheta $ with size $ d $ computing $ \saffunc{d+3} $. Then, by Theorem \ref{thm:lower-det}, we can have
	$
		N( \saffunc{d+3} ) \leq (d+1)^{d+1},
	$
	which is a contradiction and so 
	$
		\saffunc{d+3} \notin \dsizetheta{d}				
	$.
	
	Moreover, by Theorem \ref{thm:det-alg-2saf}, we have that $ \saffunc{d+3} $ is in  $ \dsize{13d+43} $.
	 
	Lastly, the relation between $ d $ and $ n $ can be followed from the definition of $ \saffunc{d} $  by also taking into account the parameter $ (13d+43) $. For $ \saffunc{t} $, we have inequality $ 2t(2t+ \lceil \log(2t) \rceil) <n $. For $ t \geq 2 $, we always have $ t \geq \lceil \log(2t) \rceil $. Therefore, the following inequality also works $ 2t(3t) < n $, which gives us that $ t < \sqrt{\frac{n}{6}}  $. 
\end{proof}

\begin{theorem}
\label{thm:hieararchy-non}
	Let $ d: \mathbb{N} \rightarrow \mathbb{N} $ be a function satisfying $ 121 < 13d(n)+ 4 < \sqrt{\frac{n}{6}} $. Then, 
	$
		\nsizetheta{ \lfloor \sqrt{d} \rfloor } \subsetneq \nsizetheta{13d+4}.
	$
\end{theorem}
\begin{proof}
	By Theorem \ref{thm:lower-2saf} we have 
	$
		2^{\log_2 d (d-2)} = d^{d-2} \leq N( \saffunc{d} ).
	$ 
	For $ d > 9 $, we can follow that
	\[
		2^{(\sqrt{d}+1)^2} < 2^{\log_2 d (d-2)}.
	\]
	Suppose that there is a $ \natheta $ with size $ \lfloor \sqrt{d} \rfloor $ computing $ \saffunc{d} $. Then, by Theorem \ref{thm:lower-non}, we can have
	$
		N( \saffunc{d} ) \leq 2^{(\lfloor \sqrt{d} \rfloor +1)^2},
	$
	which is a contradiction and so
	$
		\saffunc{d} \notin \nsizetheta{ \lfloor \sqrt{d} \rfloor }
	$.
	 
	Moreover, by Theorem \ref{thm:det-alg-2saf}, we have that $ \saffunc{d} $ is in  $ \dsize{13d+4} $.
	
	The relation between $ d $ and $ n $ can be followed similar to the previous proof and the condition $ d>9 $.
\end{proof}

\begin{theorem}
	\label{thm:hieararchy-2dfa}
    Let $ d: \mathbb{N} \rightarrow \mathbb{N} $ be a function satisfying $ 1330 < \ceil{11d(n)\log d(n)}  < \sqrt{\frac{n}{12}} $. Then, 
	$ \dfasize{d-3} \subsetneq \dfasize{ \ceil{11d\log d}  } $.
\end{theorem}
\begin{proof}
	By Theorem \ref{thm:lower-2usaf}, we have
	$
		(d-2)^{d-2} < d^{d-2} \leq N( \usaffunc{d} )
	$ and so we can also have 
	\[
		(d-2)^{d-2} <R_n(\usaflang{d}).
	\]
	Suppose that there is a 2DFA with size $ d-3 $ recognizing $ \usaflang{d} $. Then, by Corollary \ref{cor:lower-2dfa}, we can have
	$
		R_n(\usaflang{d}) \leq (d-2)^{d-2},
	$
	which is a contradiction and so
	$
		\usaflang{d} \notin \dfasize{d-3}.
	$
	
	By Theorem \ref{thm:det-alg-2usaf} and for  $ d > 26 $, we have $ \usaflang{d} \in \dfasize{ \ceil{11d\log d} } $.
	
	The relation between $ d $ and $ n $ can be followed by the definition of $ \usaffunc{t} $: $ 4t(2t+\ceil{\log(2t)}) < n $. For $ t \geq 2 $, $ t \geq \ceil{\log(2t)} $ and so we can also use the inequality $ 4t(3t) < n  $, which implies  $ t \leq \sqrt{\frac{n}{12}} $.
\end{proof}

\begin{theorem}
\label{thm:hierarchy-2nfa}
Let $ d: \mathbb{N} \rightarrow \mathbb{N} $ be a function satisfying $ 1330 < \ceil{11d\log d}  < \sqrt{\frac{n}{12}} $. Then,
	$ \nfasize{ \lfloor \sqrt{d} \rfloor   } \subsetneq \nfasize{ \ceil{11 d log d} } $.
\end{theorem}
\begin{proof}
	By using the facts given in the two above proofs, for $ d >9 $, we know that 
	$
		\usaflang{d} \notin \nfasize{ \lfloor \sqrt{d} \rfloor }
	$
	and
	$
		\usaflang{d} \in \dfasize{ \ceil{11d\log d} }.
	$	 
\end{proof}

\begin{theorem}
	\label{thm:hieararchy-pro}
	Let $ d: \mathbb{N} \rightarrow \mathbb{N} $ be a function satisfying $ 30 < 13 d(n) + 4 < \sqrt{\frac{n}{6}} $ and $T \geq 256$ be the expected running time to finish computation. Then
	\[
		\psizetheta{ \Floor{ \frac{\sqrt{d}}{32 \log T} } } \subsetneq \psizetheta{ 13d+4 },
	\]
	where the error bound is at least $ \frac{1}{5} $ for the classes.
\end{theorem}
\begin{proof}
	By using  Theorem \ref{thm:lower-2saf}, we can obtain the following for $ d>2 $.
	\[
		\paran{ \frac{32 \log(T) \sqrt{d} }{32 \log (T)} }^{ \paran{ \frac{\sqrt{d}}{32 \log T} +1 }^2 }
		<
		\paran{\frac{32 \log(T) d }{32 \log (T)}}^{d-2} 
		= d^{d-2} \leq N(\saffunc{d}).
	\]
	Then $ \saffunc{d} $ cannot be solved by a $ \patheta $ with size $ \floor{ \frac{ \sqrt{d} }{32 \log T} } $. Otherwise, by Corollary \ref{cor:lower-pro}, we obtain the following contradiction. 
	\[
		N( \saffunc{d} ) < \paran{ 32 \log T \Floor{ \frac{\sqrt{d}}{32 \log T} } } ^ { \paran{ \floor{ \frac{\sqrt{d}}{32 \log T}} + 1 }^2 } .
	\]
	By Theorem \ref{thm:det-alg-2saf}, we have that $ \saffunc{d} $ is in  $ \dsize{13d+4} $.	
	
	The relation between $ d $ and $ n $ can be followed similar to the previous proofs and the condition $ d>2 $.
\end{proof}

\section{Incomparability results}
\label{sec:incomparability}

\newcommand{\eq}{\mathtt{EQ}}
In this section, we give evidences how shuffling can reduce the size of models. For this purpose, we use the well-known {\em Equality function} 
\[
	\eq(X)=\bigvee_{0}^{\lfloor n/2\rfloor-1} x_i=x_{i+\lfloor n/2\rfloor}.
\]
The nice property of $ \eq(X) $ for our purpose is that $N^{id}(\eq)= 2^{\lfloor n/2\rfloor}$ and $N^{\theta}(\eq)\leq 4$ for $\theta=(0,\lfloor n/2\rfloor,1, \lfloor n/2\rfloor+1,\dots)$. Therefore, we can follow that $ \mathsf{\eq} \in \dsizetheta{4} $. 

Suppose that there is a $ \na $ with size $ \lfloor \sqrt{\frac{n}{2}} \rfloor - 2  $ solving $ \eq $. Then, by Theorem \ref{thm:lower-non}, we have
\[
	N^{id}(\eq) \leq 2^{ (  \lfloor \sqrt{\frac{n}{2}} \rfloor -1 )^2 },
\]
which is a contradiction. Therefore, $ \eq \notin \nsize{ \lfloor \sqrt{\frac{n}{2}} \rfloor - 2 } $. We use this lower bound also for deterministic case.

Let $ d = d(n) $ be a function such that $ d \leq \lfloor \sqrt{\frac{n}{2}} \rfloor - 2 $, and $ d' $ be a function such that $ 4 \leq d' < d $. Then, due to the fact given for function $ \eq $, we can immediately follow that
\[
	\dsizetheta{d'} \notin \dsize{d}~~ \mbox{ and } ~~\nsizetheta{d'}  \notin \nsize{d} .
\]
On the other hand, due to the results given in Section \ref{sec:hierarchy}, the shuffling classes cannot contain the non-shuffling classes under certain size bounds, which leads us to our incomparability results.

\begin{theorem}
	As a further restriction, if $ 13d + 43 < \sqrt{\frac{n}{6}}$, then for $ 94 \leq 13d'+43 < d $, the classes $ \dsizetheta{d'} $ and $ \dsize{d} $ are not comparable.
\end{theorem}
\begin{proof}
	By the proof of Theorem \ref{thm:hieararchy-det}, we know that $ \dsizetheta{d'} $ does not have a function that is in $ \dsize{13d'+43} $. Since $ 13d'+43 < d $, this function is also in $ \dsize{d} $.
\end{proof}

\begin{theorem}
	As a further restriction, if $  121 < 13d+ 4 < \sqrt{\frac{n}{6}} $, then for $ 58 \leq 13d'+4 < d $, the classes $ \nsizetheta{\floor{ \sqrt{d'} }} $ and $ \nsize{d} $ are not comparable.
\end{theorem}
\begin{proof}
	By the proof of Theorem \ref{thm:hieararchy-non}, we know that $ \nsizetheta{\floor{ \sqrt{d'} }} $ does not have a function that is in $ \dsize{ 13d'+4 } $. Since $ 13d'+4 < d $, this function is also in $ \nsize{d} $.
\end{proof}

Any $ \pa $ with size $ d $ satisfying $ \paran{ 32d \log T }^{ (d+1)^2 } < 2^{\floor{\frac{n}{2}}} $ cannot solve $ \eq $ due to Theorem \ref{thm:lower-pro}. We bound $ T < 2^{2^d} $. Then, for $ d \geq 2 $, we can follow
\[
	\paran{ 32d \log T }^{ (d+1)^2 } < \paran{ 2^5 2^{\log_2 d} 2^d }^{(d+1)^2} \leq \frac{2^{d^3}}{2} < \frac{ 2^{ \frac{n}{2}} }{2}  < 2^{\floor{\frac{n}{2}}}.
\]
Then, for $ 4 \leq d' \leq d $ where $ d < \sqrt[3]{\frac{n}{2}} $, $ \psizetheta{d'} \notin \psize{d} $. Now, we can state the result for probabilistic case similar to the other cases.
\begin{theorem}\label{p-th2}
	As a further restriction, if $ 30 < 13 d + 4 < \sqrt[3]{\frac{n}{2}} $, then for $ 4 \leq 13d'+4 < d $ and $ 256 \leq  T < 2^{2^{d'}} $ be the expected running time to finish computation, the classes $ \psizetheta{ \Floor{ \frac{\sqrt{d'}}{32\log T}  } } $ and $ \psize{d} $ are not comparable.	
\end{theorem}

\subparagraph*{Acknowledgements.}

The authors thank to A. Ambainis, A. Rivosh, K. Pr\={u}sis, and J. Vihrovs for useful discussions and comments.

The work is partially supported by ERC Advanced Grant MQC. The work is performed according to the Russian Government Program of Competitive Growth of Kazan Federal University



\end{document}